\documentclass[a4paper,fleqn]{cas-dc}

\usepackage[numbers]{natbib}

\usepackage[justification=centering]{caption}

\usepackage{amsfonts,amsmath,amssymb,amsthm}

\usepackage{threeparttable}

\usepackage{algorithm} 
\usepackage[noend]{algpseudocode} 

\usepackage{subcaption}

\def\tsc#1{\csdef{#1}{\textsc{\lowercase{#1}}\xspace}}
\tsc{WGM}
\tsc{QE}
\tsc{EP}
\tsc{PMS}
\tsc{BEC}
\tsc{DE}

\newtheorem{Definition}{Definition}
\newtheorem{Theorem}{Theorem}

\begin{document}
\let\WriteBookmarks\relax
\def\floatpagepagefraction{1}
\def\textpagefraction{.001}
\shorttitle{Achieve Efficient Position-Heap-based Privacy-Preserving Substring-of-Keyword Query over Cloud}
\shortauthors{Fan~Yin et~al.}

\title [mode = title]{Achieve Efficient Position-Heap-based Privacy-Preserving Substring-of-Keyword Query over Cloud}                      



\author[1,2]{Fan~Yin}
\ead{yinfan519@gmail.com}


\address[1]{The Information Security and National Computing Grid Laboratory, Southwest Jiaotong University, Chengdu, China 611756}

\author[2]{Rongxing~Lu}[orcid=0000-0001-5720-0941]
\cormark[1]
\ead{rlu1@unb.ca}

\author[2]{Yandong~Zheng}
\ead{yzheng8@unb.ca}


\address[2]{The Canadian Institute for Cybersecurity,  Faculty of Computer Science, University of New Brunswick, Fredericton, Canada E3B 5A3}

\author[3]{Jun~Shao}
\ead{chn.junshao@gmail.com}

\address[3]{School of Computer and Information Engineering, Zhejiang Gongshang University, Hangzhou, China 310018}

\author[4,5]{Xue~Yang}
\ead{yang.xue@sz.tsinghua.edu.cn}
\address[4]{The Tsinghua Shenzhen International Graduate School, Tsinghua University, Shenzhen, China 518055}
\address[5]{The PCL Research Center of Networks and Communications, Peng Cheng Laboratory, Shenzhen, China 518055}

\author[1]{Xiaohu~Tang}
\ead{xhutang@swjtu.edu.cn}

\cortext[cor1]{Corresponding author}


\begin{abstract}
The cloud computing technique, which was initially used to mitigate the explosive growth of data, has been required to take both data privacy and users' query functionality into consideration. Symmetric searchable encryption (SSE) is a popular solution to supporting efficient keyword queries over encrypted data in the cloud. However, most of the existing SSE schemes focus on the exact keyword query and cannot work well when the user only remembers the substring of a keyword, i.e., substring-of-keyword query. This paper aims to investigate this issue by proposing an efficient and privacy-preserving substring-of-keyword query scheme over cloud. First, we employ the position heap technique to design a novel tree-based index to match substrings with corresponding keywords. Based on the tree-based index, we introduce our substring-of-keyword query scheme, which contains two consecutive phases. The first phase queries the keywords that match a given substring, and the second phase queries the files that match a keyword in which people are really interested. In addition, detailed security analysis and experimental results demonstrate the security and efficiency of our proposed scheme.
\end{abstract}




\maketitle

\section{Introduction}

The rapid development of information techniques, e.g., internet of things, smart building, etc., has been promoting the explosive growth of the data. According to IBM Marketing Cloud study \cite{cloud10}, more than 90\% of data on the internet has been created since 2016. In order to mitigate the local storage and computing pressure, an increasing number of individuals and organizations tend to store and process their data in the cloud. However, since the cloud server may not be fully trustable, those data with some sensitive information (e.g., electronic health records) have to be encrypted before being outsourced to the cloud \cite{ZhengLLSYC19}. Although the data encryption technique can preserve data privacy, it also hides some critical information such that the cloud server cannot well support some users' query functionality over the encrypted data, e.g., keyword query, which returns the collection of files containing some specific queried keywords. In order to address the challenge, the concept of symmetric searchable encryption (SSE)~\cite{SongWP00} was introduced, which enables the cloud server to search encrypted files in a very efficient way.

Over the past years, in order to improve the keyword query efficiency, a variant of secure keyword-based index techniques have been designed to match the keywords with corresponding files, such as inverted index~\cite{CurtmolaGKO06, CashJJKRS13, Cash2014}, tree-based index~\cite{goh2003, yin2019, shao2019}, etc. Since the current keyword-based index techniques are built with exact keywords, the existing SSE schemes can only support exact keyword query, i.e., the queried keyword must be exactly the same keyword stored in cloud. 

{However, in practice, it is quite common that a user only remembers a fragment/substring of a keyword rather than the exact keyword and expects to achieve a substring-of-keyword query, i.e., \textit{the user first queries some candidate keywords containing a substring to help him/her complete the queried keyword and then queries files that match the queried keyword.} Considering the Google website example, it automatically returns a list of candidate keywords after users enter a fragment of the queried keyword to the search bar. This feature can help users efficiently enter the correct queried keyword before a real search. Unfortunately, most SSE schemes with the current keyword-based index techniques cannot be directly used to support the substring-of-keyword query because their indexes do not contain the substring information. Although some SSE schemes \cite{chase2015, leontiadis2018storage, mainardi2019privacy, hahn2018practical} focus on the substring query and can be used to implement substring-of-keyword query, they cannot achieve high efficiency in terms of the computational cost of query processing and the overhead of storage at the same time.}

	To address the above challenge, in this paper, we consider a fine-grained SSE scheme supporting substring-of-keyword query, which consists of two consecutive phases. The first phase, called the substring-to-keyword query, is to query a list of candidate keywords containing a given specific substring, and then the user chooses the keyword that he/she needs from candidate keywords. The second phase, called the keyword-to-file  query, is to query files that match the chosen keyword. Specifically, the main contributions of this paper are three-fold:

\begin{itemize}		
	\item First, based on the position heap technique, we design a storage-efficient index (i.e., modified position heap) to match substrings with corresponding keywords. We then use pseudo-random function and symmetric encryption scheme to encrypt this index, which can not only well support the substring-to-keyword query, but also preserve the privacy of queried substring as well as the plaintext of the keywords.
	
	
	\item Second, we proposed an efficient and privacy-preserving substring-of-keyword query scheme, which consists of a substring-to-keyword query and a keyword-to-file query. This scheme is suitable for critical applications in practice such as Google search.
	
	\item Finally, we analyze the security of our proposed scheme and conduct extensive experiments to evaluate its performance. The results show that our proposed scheme can achieve efficient queries in terms of low computational cost and communication overhead.
\end{itemize}

The remainder of the paper is organized as follows.
We formalize the system model, security model, and design goals in Section~\ref{ch2:sec:models}. Then, we introduce some preliminaries including the position heap technique \cite{EhrenfeuchtMOW11}, symmetric encryption scheme, and the security notion of substring-to-keyword query in Section~\ref{ch2:sec:preliminaries}. After that, we present our proposed scheme in Section~\ref{ch2:sec:proposed}, followed by security analyses and performance evaluation in Section~\ref{ch2:sec:security} and Section~\ref{ch2:sec:performance}, respectively. Some related works are discussed in Section \ref{ch2:sec:related}. Finally, we draw our conclusions in Section~\ref{ch2:sec:conclusions}.

\section{Models and Design Goals}\label{ch2:sec:models}
In this section, we formalize the system model, security model, and identify our design goals.

\subsection{System Model}\label{ch2:subsec:system_model}
In our system model, we consider two entities, namely a cloud server and a data user, as shown in Figure~\ref{ch2:fig:system_model}.

\begin{figure}[htbp]
	\centering
	\includegraphics[width=0.98\linewidth]{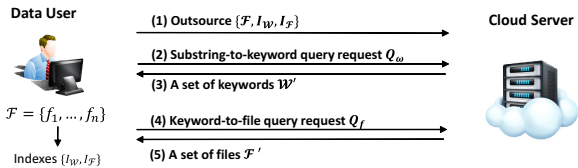}
	\caption{System model under consideration}
	\label{ch2:fig:system_model}
\end{figure}

\begin{itemize}
	{\item \underline{\textbf{Data user}}: The data user has a collection of files $\mathcal{F}=\{f_1,f_2,...,f_n\}$ and each file $f_j \in \mathcal{F}$ consists of a set of keywords from a dictionary $\mathcal{W} = \{\omega_1, \omega_2, ..., \omega_d\}$. Due to the limited storage space and computational capability, the data user intends to outsource the file collection $\mathcal{F}$ and its indices, i.e., $I_\mathcal{W}$ -- index for substring-to-keyword query, $I_\mathcal{F}$ -- index for keyword-to-file query, to the cloud server. Then, the data user launches a substring-of-keyword query with the cloud server. The substring-of-keyword query consists of two consecutive phases: a substring-to-keyword query and a keyword-to-file query. To be more specific, the data user first submits a substring-to-keyword query request $Q_\omega$ to the cloud server and retrieves a set of keywords $\mathcal{W'} \subseteq \mathcal{W}$ containing the given substring. Then, the data user chooses the queried keyword from $\mathcal{W'}$ and uses a queried keyword to submit a keyword-to-file query request $Q_f$ to retrieve a set of files $\mathcal{F}' \subseteq \mathcal{F}$ containing the queried keyword.}
	
	\item \underline{\textbf{Cloud server}}: The cloud server is considered to be powerful in storage space and computational capability. The cloud server is assumed to efficiently store file collection $\mathcal{F}$ and its indices $\{I_\mathcal{W}, I_\mathcal{F}\}$ in local. In addition, the cloud server will process two types of query requests: {substring-to-keyword query} request $Q_\omega$ and keyword-to-file query request $Q_f$. For the former, the cloud server conducts search operation in the index $I_\mathcal{W}$ and responds a set of keywords $\mathcal{W}' \subseteq \mathcal{W}$; For the latter, the cloud server conducts search operation in the index $I_\mathcal{F}$ and responds a set of files $\mathcal{F}' \subseteq \mathcal{F}$.
\end{itemize}


\subsection{Security Model}
In our security model, the data user are considered as trusted, while the cloud server is assumed as \emph{honest-but-curious}, which means that the cloud server will i) honestly execute the query processing, return the query results without tampering it, and ii) curiously infer as much sensitive information as possible from the available data. The sensitive information could include the files $\mathcal{F}$, the indices $\{I_\mathcal{W}, I_\mathcal{F}\}$, the substring-to-keyword query request $Q_\omega$, and the keyword-to-file query request $Q_f$.

Note that, since we focus on the efficiency and confidentiality of the proposed scheme, other active attacks on data integrity and source authentication are beyond the scope of this paper and will be discussed in our future work. 

\subsection{Design goals}
In this work, our design goal is to achieve an efficient and privacy-preserving substring-of-keyword query scheme. In particular, the following three requirements should be achieved.

\begin{itemize}
	\item \textit{Privacy preservation}. In the proposed scheme, all the data obtained by the cloud server, i.e., $\{\mathcal{F}, I_\mathcal{W}, I_\mathcal{F}, Q_\omega,$ $Q_f\}$, should be privacy-preserving during the outsourcing, query, and update phases. Formally, the proposed scheme needs to satisfy security definition~\ref{def:security}.
	
	\item \textit{Efficiency}. In order to achieve the above privacy requirement, additional computational cost and storage overhead will inevitably be incurred. Therefore, in this work, we also aim to reduce the computational cost and communication overhead to be linear with the length of the queried substring.
	
	\item \textit{Dynamics}. Update operations should be efficiently and securely supported after the initial outsourcing.
\end{itemize}

\section{Preliminary}
\label{ch2:sec:preliminaries}
In this section, we recall some preliminaries including the position heap~\cite{EhrenfeuchtMOW11}, the symmetric encryption scheme, and the security notion of {substring-to-keyword} query, which will be served as the basis of our proposed scheme.

\subsection{The (Original) Position Heap Technique} \label{ch2:subsec:positionheap}	

\begin{figure}[h]
	\centering
	\includegraphics[width=0.98\linewidth]{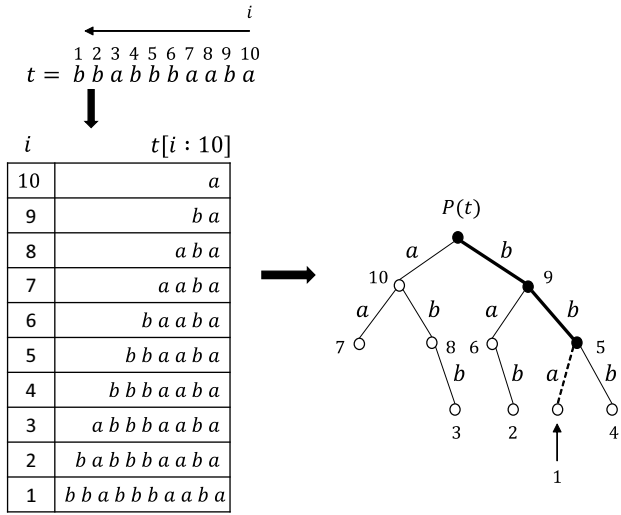}
	\captionsetup{justification=justified}
	\caption{{An example of building position heap $P(t)$ for string $t = bbabbbaaba$. The solid edges in $P(t)$ reflect the insertion path for suffix $t[1:10]$.}}
	\label{ch2:fig:example_of_building_position_heap}
\end{figure}

{Intuitively speaking, the (original) position heap $P(t)$ is a trie built from all the suffixes of $t$ and can be used to achieve efficient substring search for $t$. To construct the position heap $P(t)$ from a string $t = c_1c_2...c_m$, a set of suffixes $t[i:m] = c_i ... c_m\ (i \in [m, ..., 1])$ are chosen and inserted to the $P(t)$, which is initialized as a root node. To do this, for each suffix $t[i:m]\ (i \in [m, ..., 1])$, its longest prefix $t[i:j]\ (i \le j \le m)$ that is already represented by a path in $P(t)$ is found and a new leaf child is added to the last node of this path. The new leaf child is labeled with $i$ and its edge is labeled with $t[j+1]$ (see Figure~\ref{ch2:fig:example_of_building_position_heap}). 
	Compared to other data structures to achieve substring search, such as suffix tree \cite{chase2015} and suffix array \cite{leontiadis2018storage}, the position heap \cite{EhrenfeuchtMOW11} can achieve high efficiency in both storage and query time.

	
	
	In the following, we formally describe the $\mathtt{PHBuild}$ and $\mathtt{PHSearch}$ algorithms of the position heap. Note that, we consider each node in the position heap stores two attributes: $edge$ and $pos$, where the former stores the label of the node's edge and the latter stores the label of the node.}

\begin{algorithm}[htbp]
	\caption{Build a position heap $P(t)$ for the string $t=c_1c_2...c_m$} \label{alg:build}
	\begin{algorithmic}[1]
		\State {initialize position heap $P(t)$ as a root node $R$, where $R.edge = Null$ and $R.pos = Null$;}
		\For {each $i$ in $[m, m-1, ..., 1]$}
		\State {$N = R$;}
		\For {each $j$ in $[i, i+1, ..., m]$}
		\State {find the child $N'$ of $N$, where $N'.edge = c_j$;}
		\If {$N'$ does exist}
		\State {$N = N'$}
		\Else
		\State {$j = j - 1$;}
		\State {break;}
		\EndIf	
		\EndFor
		\State {insert a new child node $N'$ to the $N$;}
		\State {$N'.edge = c_{j+1}, N'.pos = i$;}
		\EndFor
		\State {return $P(t)$;}
	\end{algorithmic}
\end{algorithm}

\begin{algorithm}[H]
	\caption{Search substring $s$ in a position heap $P(t)$, where $s = s_1s_2...s_l$ and $t = c_1c_2...c_m$} \label{alg:Search}
	\begin{algorithmic}[1]
		\State {initial empty sets $L_1$ and $L_2$;}	
		\State {let $N$ be the root node of the $P(t)$;}
		\For {each $i$ in $[1, 2, ..., l]$}
		\State {find the child $N'$ of $N$, where $N'.edge = s_i$;}
		\If {$N'$ does exist}
		\If {$i = l$}
		\State {$L_2.add(N'.pos)$;}
		\For {each descendant $X$ of $N'$}
		\State {$L_2.add(X.pos)$;}
		\EndFor
		\Else
		\State {$L_1.add(N'.pos)$;}
		\EndIf
		\State {$N = N'$;}
		\Else
		\State {break;}
		\EndIf
		\EndFor
		\For {each $i$ in $L_1$}
		\If {$c_ic_{i+1}...c_{i+l-1}$ is not equal to $s_1s_2...s_l$}
		\State {$L_1.remove(i)$;}
		\EndIf
		\EndFor
		\State {return $L_1 \cup L_2$;}
	\end{algorithmic}
\end{algorithm}

\subsubsection{$\mathtt{PHBuild}$ Algorithm}
Given a string $t=c_1c_2...c_m$, the $\mathtt{PHBuild}$ (i.e., Algorithm~\ref{alg:build}) first initializes position heap $P(t)$ as a root node. Then, it visits the $t$ from the right to left and inserts each suffix $t[i:m]\ (i \in [m,..,1])$ to the position heap $P(t)$.
In particular, for each suffix $t[i:m]$, the algorithm first finds its longest prefix $t[i:j]\ (i \le j \le m)$ that is already represented by a path in $P(t)$ (lines 4-10). Assume the last node of this path is $N$. Then the algorithm appends a new leaf child $N'$ to the $N$, where $N'.edge = c_{j+1}$ and $N'.pos = i$ (lines 11-12). Figure~\ref{ch2:fig:example_of_building_position_heap} depicts an example to build such a position heap for a string $t = bbabbbaaba$. During the insertion for suffix $t[1:10]$, this algorithm finds its longest prefix $t[1:2]$ represented by the solid path and appends a new leaf child $N'$ to the last node of the solid path, where $N'.edge = a$ and $N'.pos = 1$.	

\subsubsection{$\mathtt{PHSearch}$ Algorithm}
Given a substring $s$ and a position heap $P(t)$, the $\mathtt{PHSearch}$ (i.e., Algorithm~\ref{alg:Search}) is supposed to find all the positions in $t$ that is occurrences of $s$. The time complexity of this algorithm is $O(|s|^2 + d_s)$, where $|s|$ is the length of the queried substring and $d_s$ is the number of matching positions. The details are as follows:
\begin{itemize}
	\item The algorithm first finds the longest prefix $s'$ of $s$ that can be represented by a path in $P(t)$ and denote this path as search path. Then the algorithm {lets} $L_1$ be the set of $pos$ stored in the intermediate nodes along the search path and $L_2$ be the set of $pos$ stored in the descendants of the last node of the search path (lines 3-14). In particular, if $s' \neq s$, the $pos$ stored in the last node of the search path is included in $L_1$. Otherwise, it is included in $L_2$. 
	\item {After completing the previous step, elements in $L_2$ must be the matching positions, and elements in $L_1$ may or may not be the matching positions. Therefore, the algorithm reviews each position $i \in L_1$ in the string $t$ to filter out unmatching positions and remove them from the $L_1$.} Finally, this algorithm returns $L_1 \cup L_2$ (lines 15-17).
\end{itemize}
Take an example with Figure~\ref{ch2:fig:example_of_building_position_heap}. Given a substring $s=bb$, the $\mathtt{PHSearch}$ algorithm finds its longest prefix $bb$ corresponding to the solid path. In this way, $L_1$ and $L_2$ are equal to $\{9\}$ and $\{5, 1, 4\}$. Then, this algorithm reviews the string $t$ and makes sure $i=9 \in L_1$ is not an occurrence of $s$. Therefore, the position $9$ is removed from the $L_1$, and $L_1$ is an empty set now. Finally, this algorithm returns all the $pos$ in $L_1 \cup L_2 = \{5, 1, 4\}$.

\subsection{Symmetric Key Encryption Scheme}
{A symmetric key encryption scheme (SKE) is a set of three polynomial-time algorithms $(Gen, Enc, Dec)$ such that $Gen$ takes a security parameter $\lambda$ and returns a secret key $K$; $Enc$ takes a key $K$ and a message $M$, then returns a ciphertext $C$; $Dec$ takes a key $K$ and a ciphertext $C$, then returns $M$ if $K$ was the key under which $C$ was produced. In this work, we consider a SKE is indistinguishable under chosen plaintext attack (IND-CPA)~\cite{katz2014introduction}, which guarantees the ciphertext does not leak any information about the plaintext even an adversary can query an encryption oracle. We note that common private-key encryption schemes such as AES in counter mode satisfy this definition.}

\subsection{Security Definition of Substring-to-Keyword Query} \label{ch2:sec:Security_Definition}
In this subsection, we follow the security definition in \cite{CurtmolaGKO06} to formalize the simulated-based security definition of substring-to-keyword query scheme by using the following two experiments: $Real_{\mathcal{A,C}}(\lambda)$ and $Ideal_{\mathcal{A,S}}(\lambda)$. In the former, the adversary $\mathcal{A}$, who represents the cloud server, executes the proposed scheme with a challenger $\mathcal{C}$ that represents the data user. In the latter, $\mathcal{A}$ also executes the proposed scheme with a simulator $\mathcal{S}$ who simulates the output of the challenger $\mathcal{C}$ through the leakage of the proposed scheme. The leakage is parameterized by a leakage function collection $\mathcal{L}=(\mathcal{L}_{O}, \mathcal{L}_{Q}, \mathcal{L}_{U})$, which describes the information leaked to the adversary $\mathcal{A}$ in the outsourcing phase, query phase, and update phase respectively. If any polynomial adversary $\mathcal{A}$ cannot distinguish the output information between the challenger $\mathcal{C}$ and the simulator $\mathcal{S}$, then we can say there is no other information leaked to the adversary $\mathcal{A}$, i.e., the cloud server, except the information that can be inferred from the $\mathcal{L}$. More formally,
\begin{itemize}
	\item $Real_{\mathcal{A,C}}(1^\lambda) \to b \in \{0,1\}$: Given a keyword dictionary $\mathcal{W}$ chosen by the adversary $\mathcal{A}$, the challenger $\mathcal{C}$ outputs encrypted index $I$ by following the outsourcing phase of the proposed scheme. Then, $\mathcal{A}$ can adaptively send a polynomial number of query requests (or update requests) to the $\mathcal{C}$, which outputs corresponding encrypted query requests (or encrypted update requests). Eventually, $\mathcal{A}$ returns a bit $b$ as the output of this experiment.
	
	\item $Ideal_{\mathcal{A,S}}(1^\lambda) \to b \in \{0,1\}$: Given the leakage function $\mathcal{L}_{O}$, the simulator outputs simulated encrypted index $\overline{I}$. Then, for each query request (or update request), the adversary $\mathcal{A}$ sends its leakage function $\mathcal{L}_{Q}$ (or $\mathcal{L}_{U}$) to the simulator $\mathcal{S}$, which generates the corresponding simulated encrypted query request (or encrypted update request). Eventually, $\mathcal{A}$ returns a bit $b$ as the output of this experiment.
\end{itemize}

\begin{Definition} \label{def:security}
	A substring-to-keyword query scheme is $\mathcal{L}$-secure against adaptive attacks if for any probabilistic polynomial time adversary $\mathcal{A}$, there exists an efficient simulator $\mathcal{S}$ such that\\
	$$
	|{\rm Pr}[Real_{\mathcal{A, C}}(\lambda) \to 1]-{\rm Pr}[Ideal_{\mathcal{A,S,L}}(\lambda) \to 1]| \le negl(\lambda).
	$$
\end{Definition}

\begin{figure*}
	\centering
	\includegraphics[width=0.98\linewidth]{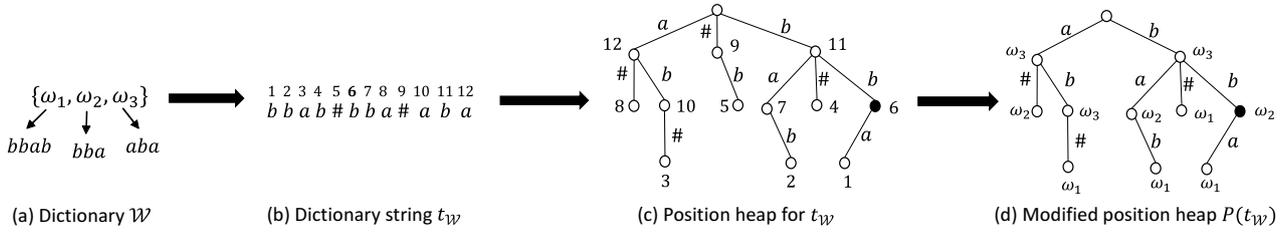}
	\captionsetup{justification=justified}
	\caption{An example of building a modified position heap for a dictionary $\mathcal{W}$. (a) $\mathcal{W} = \{\omega_1, \omega_2, \omega_3\}$ is a dictionary, where $\omega_1 = bbab$, $\omega_2 = bba$, and $\omega_3 = aba$. (b) To get dictionary string $t_\mathcal{W}$, concatenate all the keywords in $\mathcal{W}$ with character $\#$. (c) Build an original position heap for $t_\mathcal{W}$. (d) For each node $N$, replace its $N.pos$ with the corresponding keyword, called $N.keyword$. At the same time, remove useless paths from the $P(t_\mathcal{W})$.}
	\label{ch2:fig:positionheapforkeywords}
\end{figure*}

\section{Our Proposed Scheme}
\label{ch2:sec:proposed}
In this section, we will present our substring-of-keyword query scheme. Before delving into the details, we first introduce a modified position heap for keyword dictionaries, which is a basic building block of our proposed scheme.

\subsection{Modified Position Heap for Keyword Dictionaries}\label{ch2:subsec:modifiedpositionheap}	
In order to process substring-to-keyword query efficiently, we first design a modified position heap to index all the keywords in a dictionary, which mainly consists of two algorithms: i) $\mathtt{MPHBuild}$ Algorithm; ii) $\mathtt{MPHSearch}$ Algorithm.

\subsubsection{$\mathtt{MPHBuild}$ Algorithm}
Given a dictionary $\mathcal{W} = \{\omega_1, \omega_2, ..., \omega_d\}$, the $\mathtt{MPHBuild}$ algorithm first transforms the dictionary $\mathcal{W} = \{\omega_1, \omega_2, ..., \omega_d\}$ to a string $t_{\mathcal{W}} = \omega_1 || \# || \omega_2 || \# ... \# || \omega_d$, where $||$ denotes the concatenation operation and $\#$ denotes a {separate} character that does not appear in any $\omega \in \mathcal{W}$. In the rest of this paper, we call this string $t_\mathcal{W}$ \textit{dictionary string}. Then, this algorithm follows $\mathtt{PHBuild}$ algorithm (i.e., Algorithm~\ref{alg:build}) to build an original position heap for this dictionary string $t_\mathcal{W}$ and further modifies it to modified position heap $P(t_\mathcal{W})$ as follows: i) For each node $N$, replace its $N.pos$ with the corresponding keyword in $t_{\mathcal{W}}$, called $N.keyword$. ii) At the same time, remove useless paths, whose edges starting with $\#$. Figure~\ref{ch2:fig:positionheapforkeywords} depicts an example of building the modified position heap $P(t_\mathcal{W})$ for a dictionary $\mathcal{W} =\{\omega_1, \omega_2, \omega_3\}$.

\subsubsection{$\mathtt{MPHSearch}$ Algorithm} \label{ch2:subsubsec:search}
Given a substring $s$ and a modified position heap $P(t_\mathcal{W})$, the $\mathtt{MPHSearch}$ algorithm follows the $\mathtt{PHSearch}$ algorithm (i.e., Algorithm~\ref{alg:Search}) to return all the keywords in $\mathcal{W}$ that include $s$. There are two differences between $\mathtt{PHSearch}$ and $\mathtt{MPHSearch}$: i) $\mathtt{PHSearch}$ returns a set of positions, but \\$\mathtt{MPHSearch}$ returns a set of keywords because all the $N.pos$ stored in $P(t_\mathcal{W})$ is replaced by the corresponding $N.keyword$. ii) $\mathtt{PHSearch}$ reviews each position $i \in L_1$ in the string $t$ to filter out unmatching positions, but $\mathtt{MPHSearch}$ directly returns all the keywords in $L_1$. The reason is that the cloud server, who performs $\mathtt{MPHSearch}$ algorithm, is not allowed to access to the dictionary string $t_\mathcal{W}$ to filter out unmatching keywords in $L_1$. Therefore, the cloud server returns all the keywords in $L_1$ and leave the filter operation to the data user. Fortunately, the computational cost of the filter operation, i.e., $O(|s|^2)$, is acceptable because the length of queried substring (i.e., $|s|$) is relatively small in practice.

\subsection{Description of Our Proposed Scheme}
In this subsection, we will describe our proposed privacy-preserving substring-of-keyword query scheme, which mainly consists of five phases: i) System Initialization; ii) Data Outsourcing; iii) Substring-of-keyword Query; iv) Update (Insertion); and v) Update (Deletion).

\subsubsection{System Initialization}
{Given a security parameter $\lambda$, the data user first initializes a secure pseudo-random function (PRF) $H_{k_1}: \{0,1\}^{*} \longrightarrow \{0,1\}^{\gamma}$, where $k_1$ is a $\lambda$-bit random key. Then, the data user initializes an IND-CPA secure SKE ${\varPi = (Gen, Enc, Dec)}$ and generates a secret key $k_2 = {\varPi.Gen(1^\lambda)}$.}

\subsubsection{Data Outsourcing} \label{ch2:subsubsec:outsourcing}
\begin{figure}[h]
	\centering
	\includegraphics[width=0.5\textwidth]{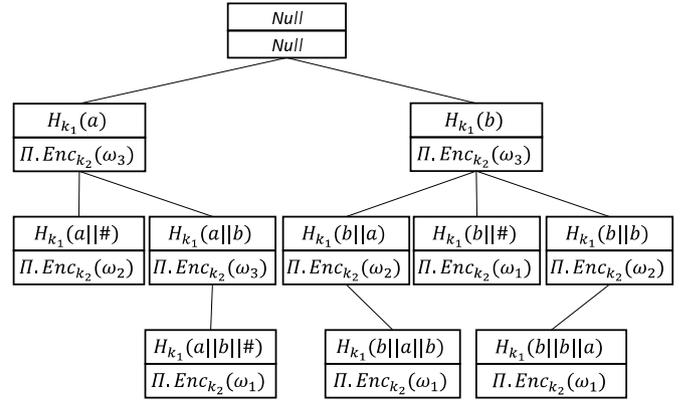}
	\captionsetup{justification=justified}
	\caption{An example of secure index $I_\mathcal{W}$, which is generated from the modified position heap $P(t_\mathcal{W})$ in Figure~\ref{ch2:fig:positionheapforkeywords}(d).}
	\label{ch2:fig:encryption_for_positionheap}
\end{figure}
Assume the data user has a file collection $\mathcal{F} = \{f_1, f_2, ...,$ $f_n\}$, where each $f_j \in \mathcal{F}$ includes a set of keywords $\mathcal{W}_j \subseteq \mathcal{W}$. The data user generates secure indices $\{I_\mathcal{W}, I_\mathcal{F}\}$ and a set of encrypted files in the following steps:

\textit{\textbf{Step 1:}} In order to support efficient substring-to-keyword query, the data user uses the $\mathtt{MPHBuild}$ algorithm, described in Section~\ref{ch2:subsec:modifiedpositionheap}, to build a modified position heap $P(t_\mathcal{W})$ for the dictionary $\mathcal{W}$.

\textit{\textbf{Step 2:}} For privacy, the data user encrypts $P(t_\mathcal{W})$ to a secure index $I_\mathcal{W}$ as follows (shown in Figure~\ref{ch2:fig:encryption_for_positionheap}):		
\begin{itemize}
	\item For each node $N$ in the modified position heap (except the root), the data user uses $\varPi.Enc_{k_2}$ to encrypt its $N.keyword$.
	\item For each node $N$ in the modified position heap (except the root), the data user concatenates each edge label, i.e., $N.edge$, along the path from the root to this node, and calculates the PRF output of the concatenation through $H_{k_1}$.
\end{itemize}
Considering the example in Figure~\ref{ch2:fig:encryption_for_positionheap}, the $I_\mathcal{W}$ is encrypted from Figure~\ref{ch2:fig:positionheapforkeywords}(d). {For each node, its keywords are encrypted through $\varPi.Enc_{k_2}$, and its edge label are transformed to a PRF output through $H_{k_1}$.}

\textit{\textbf{Step 3:}} In order to support efficient keyword-to-file query, the data user utilizes the inverted index proposed in \cite{Cash2014} to implement the index $I_\mathcal{F}$. This inverted index is implemented by a hash table, and each $<key, value>$ pair in it is the form of $<\omega, id>$, where $\omega$ is a keyword and $id$ is a file identifier.

\textit{\textbf{Step 4:}} Finally, the data user encrypts each file $f_j \in \mathcal{F}$ through $\varPi.Enc_{k_2}$ and sends these encrypted files to the cloud server with secure indices $\{I_\mathcal{W}, I_\mathcal{F}\}$.

\subsubsection{Substring-of-keyword Query} \label{ch2:subsubsec:query}
\begin{figure*}
	\centering
	\includegraphics[width=0.98\textwidth]{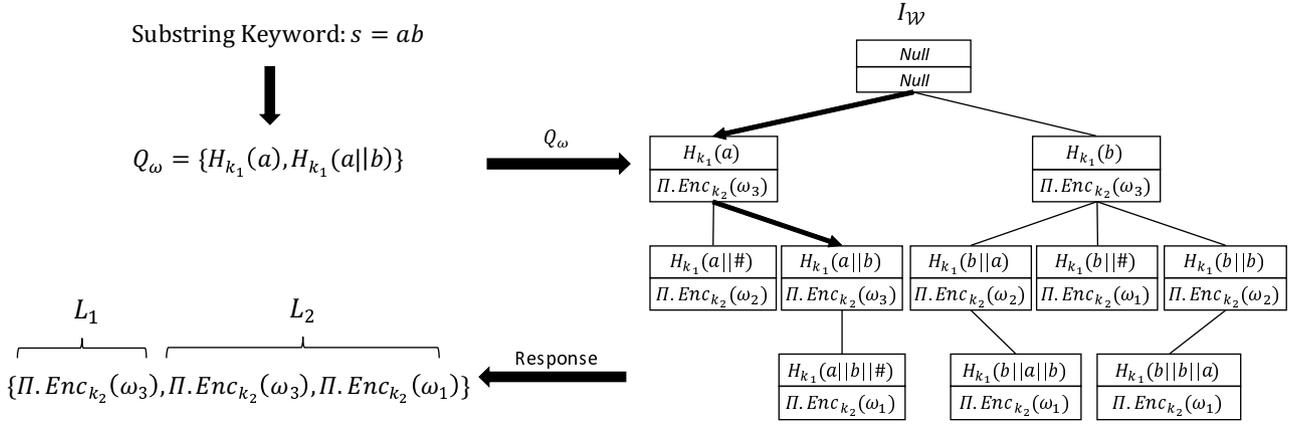}
	\captionsetup{justification=centering}
	\caption{An example of substring-to-keyword query, where the secure index $I_\mathcal{W}$ is for dictionary string $t_\mathcal{W}=bbab\#bba\#aba$ and the given substring is $s = ab$.}
	\label{ch2:fig:query}
\end{figure*}
Given a substring $s = s_1s_2...s_l$, the data user launches a substring-of-keyword query with the cloud server. The substring-of-keyword query consists of two consecutive phases: a substring-to-keyword query and a keyword-to-file query, which are described in the following steps:

\textit{\textbf{Step 1:}} First, the data user generates a {substring-to-keyword query} request $Q_\omega$ and submit it to the cloud server. To be more specific, the $Q_\omega = \{Q_1, Q_2, ..., Q_l\}$ consists of $l$ PRF outputs, where
\begin{eqnarray}
	Q_i = H_{k_1}(s_1||...||s_i), & 1 \leq i \leq l
\end{eqnarray}

\textit{\textbf{Step 2:}} After receiving the query request $Q_\omega$, the cloud server follows the $\mathtt{MPHSearch}$ algorithm, described in Section~\ref{ch2:subsec:modifiedpositionheap}, to search encrypted keywords in secure index $I_\mathcal{W}$ and returns elements in $L_1 \cup L_2$ to the data user. Figure~\ref{ch2:fig:query} depicts an example of {substring-to-keyword query}, where the given substring is $s=ab$. In this example, the data user generates $Q_\omega = \{H_{k_1}(a), H_{k_1}(a||b)\}$ and sends it to the cloud server. After receiving the $Q_\omega$, the cloud server performs $\mathtt{MPHSearch}$ to get {$L_1 = \{\varPi.Enc_{k_2}(\omega_3)\}$, $L_2 = \{\varPi.Enc_{k_2}(\omega_3), \varPi.Enc_{k_2}(\omega_1)\}$} and return $L_1 \cup L_2$ to the data user. Note that, since $\varPi.Enc$ is a randomized encryption, these encrypted keywords in $L_1 \cup L_2$ are indistinguishable for the cloud server.

\textit{\textbf{Step 3:}} After receiving the encrypted keywords, the data user first decrypts them and filters out the unmatching keywords. Then, the data user chooses a queried keyword from the matching keywords and submits a keyword-to-file query to the cloud server. Since our paper just focuses on the design of substring-to-keyword query, we directly utilize the scheme proposed in \cite{Cash2014} to implement our keyword-to-file query. According to the scheme in \cite{Cash2014}, the data user can submit efficient and privacy-preserving single keyword queries to the cloud server based on the index $I_\mathcal{F}$.

\subsubsection{Update (Insertion)} \label{ch2:subsubsec:Insertion}

\begin{figure*}
	\centering
	\includegraphics[width=0.98\textwidth]{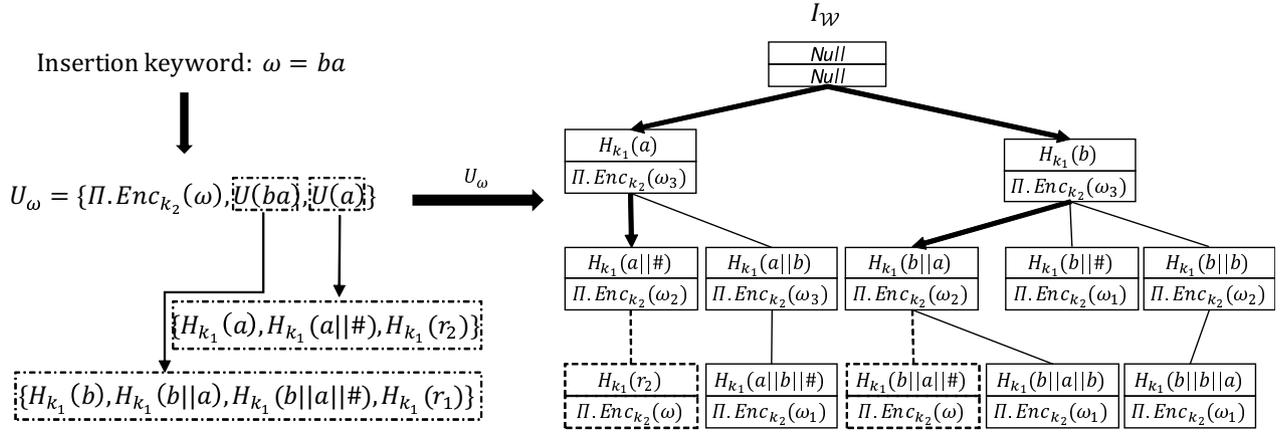}
	\captionsetup{justification=centering}
	\caption{An example of inserting keyword $\omega = ba$ to $I_\mathcal{W}$, which is the secure index for dictionary string $t_\mathcal{W}=bbab\#bba\#aba$.}
	\label{ch2:fig:Insertion}
\end{figure*}

In the proposed scheme, there are two types of insertion operations: insert {keywords} to the index $I_\mathcal{W}$ and insert {files} to the index $I_\mathcal{F}$. Since Cash et al.~\cite{Cash2014} has proposed a privacy-preserving insertion algorithm to deal with the insertion for the index $I_\mathcal{F}$, we just need to consider the insertion for the index $I_\mathcal{W}$.

Given a keyword $\omega = c_1c_2...c_z$, the data user is supposed to insert it to the index $I_\mathcal{W}$. Intuitively, assume the dictionary is $W = \{\omega_1, \omega_2, ..., \omega_d\}$ and its corresponding dictionary string is $t_{\mathcal{W}} = \omega_1 || \# || \omega_2 || \# ... \# || \omega_d$. The insertion operation will update the index $I_\mathcal{W}$ to a new version, called $I_\mathcal{W'}$, where its corresponding $t_\mathcal{W'} = \omega || \# || t_\mathcal{W}$. The details are described as follows:

\textit{\textbf{Step 1:}} The data user chooses $z$ random values $r_1,r_2,...,r_z$ to generate an update (insertion) request $U_\omega = \{\varPi.Enc_{k_2}(\omega),$ $U(c_1c_2...c_z), U(c_2...c_z), ..., U(c_z)\}$ and submit it to the cloud server. Specifically, each  $U(c_i...c_z)$ $(1 \le i \le z)$ in $U_\omega$ consists of $(z-i+3)$ PRF outputs, i.e., $\{U_i, U_{i+1}, ..., U_{z+2}\}$, where 
\begin{eqnarray}
	U_j = \left\{
	\begin{array}{ll}
		H_{k_1}(c_i||...||c_j),  & if\ i \le j \le z\\
		H_{k_1}(c_i||...||c_j||\#), & if\ j = z+1\\
		H_{k_1}(r_i), & if\ j = z+2
	\end{array}
	\right.
\end{eqnarray}
Figure~\ref{ch2:fig:Insertion} depicts an example of the insertion operation for the inserted keyword $\omega = ba$. In this example, the $U_\omega = \{\varPi.Enc_{k_2}(\omega), U(ba), U(a)\}$, where ${U(ba)} = \{H_{k_1}(b), H_{k_1}(b||a),$ $H_{k_1}(b||a||\#), H_{k_1}(r_1)\}$ and $U(a) = \{$ $H_{k_1}(a), H_{k_1}(a||\#), H_{k_1}(r_2)\}$.


\textit{\textbf{Step 2:}} After receiving the update (insertion) request $U_\omega$, for each $U(c_i...c_z)\ (1 \leq i \leq z)$ in it, the cloud server first finds its longest prefix $U_iU_{i+1}...U_h (1 \le h < z+2)$ that is already represented by a path in $I_\mathcal{W}$ and denotes this path as insertion path. Then the cloud server appends a new leaf child $N'$ to the last node of the insertion path, where $N'.edge = U_{h+1}$ and $N'.keyword = \varPi.Enc_{k_2}(\omega)$. Note that, in practice, the $h$ can not equal to $z+2$ because $U_{z+2} = H_{k_1}(r_i)$ is a random number. As shown in Figure~\ref{ch2:fig:Insertion}, the solid edges reflect the insertion paths for the $U(ba)$ and $U(a)$. 


\subsubsection{Update (Deletion)}
In the proposed scheme, there are two types of deletion operations: delete substrings from the index $I_\mathcal{F}$ and delete keywords from the index $I_\mathcal{W}$. Since Cash et al.~\cite{Cash2014} has proposed a privacy-preserving deletion algorithm for the index $I_\mathcal{F}$, we just need to consider the deletion for the index $I_\mathcal{W}$.

We implement this deletion operation by maintaining a revocation list $I_{\mathcal{W}_r}$, which is also an encrypted modified position heap, in the cloud server. Specifically, in the data outsourcing phase, the data user build a modified position heap $I_{\mathcal{W}_r}$ for an empty dictionary $\mathcal{W}_r = \{\}$ and sends the $I_{\mathcal{W}_r}$ to the cloud server with $\{I_\mathcal{W}, I_\mathcal{F}\}$. Then, to delete a keyword from the cloud server, the data user just follows the update (insertion) method in \ref{ch2:subsubsec:Insertion} to insert the keyword to the $I_{\mathcal{W}_r}$. During a substring-to-keyword query, after receiveing a substring-to-keyword query request, the cloud server performs search operations over $I_\mathcal{W}$ and $I_{\mathcal{W}_r}$ separately, and returns two result sets to the data user. Finally, the data user decrypts the two result sets and calculates the difference between them to obtain the correct keywords.

{\textbf{Correctness.} The correctness of our proposed is quite straightforward. The only issue is the collision among the edges' PRF outputs in $I_\mathcal{W}$. Since the domain size of PRF $H_{k_1}$ is $2^\gamma$, assuming that the number of nodes in $I_\mathcal{W}$ is $m$, the probability of collision is $O(\left[m \atop 2\right]/2^\gamma) = O(m^2/2^\gamma)$. So we need to choose $\gamma= \lambda + 2log(m)$ such that $O(m^2/2^\gamma) = O(1/2^\lambda)$ is negligible over the security parameter $\lambda$.}

\section{Security Analysis}
\label{ch2:sec:security}
In this paper, the proposed substring-of-keyword query scheme consists of two query schemes: a substring-to-keyword query scheme and a keyword-to-file query scheme. Since the security analysis in \cite{Cash2014} has shown that the keyword-to-file query scheme is secure, we mainly focus on the security analysis of the substring-to-keyword query scheme in this section.

\subsection{Leakage Function Collection}
{
	The leakage function collection $\mathcal{L}$ consists of three leakage functions: $\mathcal{L}_O$, $\mathcal{L}_Q$, and $\mathcal{L}_U$. Before defining them, we first give some definitions for the leakage of this scheme.
	
	\begin{Definition} \textbf{(Access Pattern)}
		Given the index $I_\mathcal{W}$, which contains a set of encrypted nodes $\{n_1, n_2, ..., n_{m}\}$, and a query request $Q_\omega$, the path pattern reveals the set of identifiers of nodes in $I_\mathcal{W}$ that are returned to the data user.
	\end{Definition}
	
	\begin{Definition} \textbf{(Query Path Pattern)}
		Given the index $I_\mathcal{W}$, which contains a set of encrypted nodes $\{n_1, n_2, ..., n_{m}\}$, and a query request $Q_\omega$, the query path pattern reveals the set of identifiers of nodes in $I_\mathcal{W}$ that are reached by the $Q_\omega$, i.e., nodes in the search path.
	\end{Definition}
	
	\begin{Definition} \textbf{(Insertion Path Pattern)}
		Given the index $I_\mathcal{W}$, which contains a set of encrypted nodes $\{n_1, n_2, ..., n_{m}\}$, and an update (insertion) request $U_\omega$, the insertion path pattern reveals the set of identifiers of nodes in $I_\mathcal{W}$ that are reached by the $U_\omega$, i.e., nodes in the insertion path.
	\end{Definition}
	
	\begin{Definition} \textbf{(Deletion Path Pattern)}
		The deletion method is implemented by a revocation list, which means the update (deletion) request is exactly the same as the update (insertion) request. Therefore, given the revocation list $I_{\mathcal{W}_r}$, which contains a set of encrypted nodes $\{n_1, n_2, ..., n_{m}\}$, and an update (deletion) request $U_\omega$, the deletion path pattern reveals the set of identifiers of nodes in $I_{\mathcal{W}_r}$ that are reached by the $U_\omega$.
	\end{Definition}
	
	Now we define the leakage functions to capture the information leakage in different phases.
}	
\subsubsection{Outsourcing Phase}
Given the index $I_\mathcal{W}$, which contains a set of encrypted nodes $\{n_1, n_2, ..., n_{m}\}$. The leakage $\mathcal{L}_O$ consists of the following information:
\begin{itemize}
	\item $m$: the size of the dictionary string $t_{\mathcal{W}}$.
	\item $\varGamma=\{(id_1, C_{id_1}), ..., (id_m, C_{id_m})\}$: the structure of index $I_\mathcal{W}$, where $id_i (1 \le i \le m)$ denotes the identifiers of encrypted node $n_i$ and $C_{id_i} (1 \le i \le m)$ denotes all the identifiers of $id_i$'s children.
\end{itemize}

\subsubsection{Query Phase}
Given the index $I_\mathcal{W}$ and a substring-to-keyword query request $Q_\omega$, the leakage $\mathcal{L}_Q$ consists of two parts: $Access\ Pattern$ and $Query\ Path\ Pattern$.

\subsubsection{Update Phase}
Given the index $I_\mathcal{W}$, revocation list $I_{\mathcal{W}_r}$, and an update request $U_\omega$, if update operation is insertion / deletion, the leakage $\mathcal{L}_U$ is $Insertion\ Path\ Pattern$ / $Deletion\ Path\ Pattern$.

\subsection{Security Proof}
We now prove the security of the substring-to-keyword query scheme based on the leakage function collection $\mathcal{L} = \{\mathcal{L}_O, \mathcal{L}_Q, \mathcal{L}_U\}$. Intuitively, we first define a simulator $\mathcal{S}$ based on the leakage function collection $\mathcal{L}$ and then analyze the indistinguishability between the output of the $\mathcal{S}$ in the ideal world and the challenger $\mathcal{C}$ (i.e., the data user) in the real world. Finally, we conclude that the proposed substring-to-keyword query scheme does not reveal any information beyond the leakage function collection $\mathcal{L}$ to the server. The details are as follows.
\begin{Theorem}
	{If the $H$ is a secure pseudo-random function (PRF) and $\varPi$ is an IND-CPA secure symmetric key encryption scheme (SKE), then our proposed scheme is $\mathcal{L}$-adaptively-secure.}
\end{Theorem}
\begin{proof}
	Based on the leakage function collection $\mathcal{L}$, we can build a simulator $\mathcal{S}$ as follows:
	\begin{itemize}
		\item Data Outsourcing: given the leakage function $\mathcal{L}_O = \{m, \varGamma\}$, the simulator $\mathcal{S}$ is supposed to generate a simulated $\overline{I_\mathcal{W}}$ (i.e., an encrypted modified position heap). Specifically, the simulator $\mathcal{S}$ first generates $m$ empty nodes and identifies each node a unique identifier from $\{id_1, ..., id_m\}$. Then the simulator $\mathcal{S}$ constructs these nodes to a tree (i.e., $\overline{I_\mathcal{W}}$) based on $\varGamma$, which means the $\overline{I_\mathcal{W}}$ has the same tree structure as $I_\mathcal{W}$. Next, for each node in the $\overline{I_\mathcal{W}}$, the simulator $\mathcal{S}$ chooses a random number $\overline{H}$ from the domain of $H$ as the PRF output of its edge and a random number $\overline{\varPi.Enc}$ from the domain of $\varPi.Enc$ as its encrypted keyword. {Since the output of $H$ and $\varPi.Enc$ are pseudo-random,} the adversary $\mathcal{A}$ cannot distinguish between the $\overline{I_\mathcal{W}}$ in the ideal world and the $I_\mathcal{W}$ in the real world.
		
		\item {Substring-to-keyword Query}: given the leakage function $\mathcal{L}_Q$ for a substring-to-keyword query request $Q_\omega$, the simulator $\mathcal{S}$ is supposed to generate a simulated encrypted substring-to-keyword query request $\overline{Q_\omega}$. Note that, in this phase, the simulator $\mathcal{S}$ not only has $\mathcal{L}_Q$ but also $\mathcal{L}_O$ and $\overline{I_\mathcal{W}}$ from the data outsourcing phase. Therefore, the simulator $\mathcal{S}$ can follow the $query\ path\ pattern$ in $\mathcal{L}_Q$ to find the search path in $\overline{I_\mathcal{W}}$ and output all the $\overline{H}$ stored in the search path as the $\overline{Q_\omega}$. {Since the output of $H$ is pseudo-random,} the adversary $\mathcal{A}$ cannot distinguish between the elements in $\overline{Q_\omega}$ and ${Q_\omega}$. At the same time, after receiving the $\overline{Q_\omega}$, the adversary $\mathcal{A}$ can use it to find matching encrypted keywords in $\overline{I_\mathcal{W}}$. Since these matching encrypted keywords in $\overline{I_\mathcal{W}}$ is encrypted through $\varPi.Enc$, the adversary $\mathcal{A}$ cannot distinguish them from the matching encrypted keywords in ${I_\mathcal{W}}$, which means the adversary $\mathcal{A}$ cannot distinguish between $\overline{Q_\omega}$ in the ideal world and $Q_\omega$ in the real world.
		
		\item Update: given the leakage function $\mathcal{L}_U$ for an update (insertion / deletion) request $U_\omega$, the simulator $\mathcal{S}$ is supposed to generate a simulated encrypted update request $\overline{U_\omega}$. Note that, in this phase, the simulator $\mathcal{S}$ not only has $\mathcal{L}_U$ but also $\mathcal{L}_O$ and $\overline{I_\mathcal{W}}$ / $\overline{I_{\mathcal{W}_r}}$ from the data outsourcing phase. Therefore, the simulator $\mathcal{S}$ can follow the $insertion\ path\ pattern$ / $deletion\ path\ pattern$ in $\mathcal{L}_U$ to find the insertion paths in $\overline{I_\mathcal{W}}$ / $\overline{I_{\mathcal{W}_r}}$ and output all the $\overline{H}$ stored in these insertion paths as the $\overline{U_\omega}$. {Since the output of $H$ is pseudo-random,} the adversary $\mathcal{A}$ cannot distinguish between $\overline{U_\omega}$ in the ideal world and ${U_\omega}$ in the real world.		
	\end{itemize}
	
	In summary, as the adversary $\mathcal{A}$ cannot distinguish between the outputs from the simulator $\mathcal{S}$ and the challenger $\mathcal{C}$, we can conclude that our proposed substring-to-keyword query scheme is $\mathcal{L}$-adaptively-secure.
\end{proof}

\section{Performance Evaluation}
\label{ch2:sec:performance}
In this section, we evaluate the performance of our proposed scheme from both theoretical and experimental perspectives.
\subsection{Theoretical Analysis}
\begin{table}[t]
	\centering
	\caption{Comparison between ours and existing schemes} \label{ch2:tab:comparison}
	\begin{tabular}{|c|c|c|c|c|c|c|}
		\hline Scheme & Index Space & Query Time & Dynamism \\
		\hline \cite{chase2015} & $O(m)$  & $O(|s| + d_s)$ & static \\
		\hline \cite{leontiadis2018storage} & $O(m)$ & $O(|s| + d_s)$ & static \\
		\hline \cite{hahn2018practical} & $O(m)$ & $O(|s| \cdot d_{\overline{kg}})$ & dynamic \\
		\hline \cite{moataz2018substring} & $O(m)$ & $O(m)$ & static \\
		\hline \cite{mainardi2019privacy} & $O(|\Sigma| \cdot m)$ & $O(|s| + d_s)$ & static \\
		\hline Our solution & $O(m)$ & $O(|s| + d_s)$ & dynamic \\
		\hline
	\end{tabular}
	\begin{tablenotes}
		\item $m$ is the size of dataset, $|s|$ is the size of queried substring $s$, $d_s$ is the number of matching positions for $s$, $d_{\overline{kg}}$ is the average number of matching positions for a k-gram of $s$, and $|\Sigma|$ is the number of distinct characters in dataset.
	\end{tablenotes}
\end{table}
	{We perform a theoretical comparison of our proposed substring-to-keyword scheme with existing schemes (cf. Table~\ref{ch2:tab:comparison}) from three aspects: index space, query time, and dynamism. From Table~\ref{ch2:tab:comparison}, we can see that the schemes in \cite{chase2015} and \cite{leontiadis2018storage} have the same index space (i.e., $O(m)$) and query time (i.e, $O(|s| + d_s)$). However, in practice, \cite{chase2015} will consume more index space than \cite{leontiadis2018storage} due to its suffix tree index, which just stores position data in leaf nodes and does not utilize the space of inner nodes effectively. In fact, the number of nodes in the suffix tree can be up to $2m$, where $m$ is the size of the dataset. In contrast, \cite{leontiadis2018storage} utilizes Burrows-Wheeler Transformation (BWT) to build a suffix array index to support substring query, which has better storage-efficiency than the suffix tree at the cost of worse query-efficiency.
	Later, based on the scheme in \cite{leontiadis2018storage} , \cite{mainardi2019privacy} uses Private Information Retrieval (PIR) technique to protect the access pattern, which causes high index space and query time. 
	In addition to suffix tree and suffix array, there are other auxiliary data structures~\cite{hahn2018practical, moataz2018substring} can be used to support substring query. However, their query time is unacceptable in practice.
	
	Compared with these existing schemes, our proposed \\substring-to-keyword scheme can achieve high storage-efficiency and query-efficiency at the same time. In specific, our scheme can achieve $O(m)$ complexity for index space and $O(|s| + d_s)$ complexity for query time, which are the same as \cite{chase2015, leontiadis2018storage} and better than \cite{hahn2018practical, moataz2018substring, mainardi2019privacy}. In addition, our proposed scheme can support dynamic datasets, which cannot be supported by \cite{chase2015, leontiadis2018storage}. Further, due to the use of position heap technique, which is storage-efficient than the suffix tree and query-efficient than the suffix array, our proposed scheme consumes less index space than \cite{chase2015} and less query time than \cite{leontiadis2018storage} in practice, which will be shown in the next subsection.
}

\subsection{Experimental Analysis}

In this subsection, we evaluate the computational cost and storage overhead of the proposed substring-to-keyword scheme in terms of three phases: local data outsourcing, substring-to-keyword query, and update. Specifically, we implemented the proposed scheme in C++ (our code is open source~\cite{Yin19}) and {conducted experiments on a 64-bit machine with an Intel Core i5-8400 CPU at 2.8GHZ and 2GB RAM, running CentOS 6.6}. We utilized the OpenSSL library for the entailed cryptographic operations, where the $H$ and $\varPi$ are instantiated using HMAC-SHA-256 and AES-512-CBC in the OpenSSL library, respectively. Note that, we implemented the data user and the cloud server on the same machine, which means there is no network delay between them. The underlying dataset (i.e., the dictionary $\mathcal{W}$) in our experiment was extracted from 29,378 articles from Wikivoyage~\cite{wiki}, and it contains 40,205 distinct keywords in total. The length distribution of the keywords in $\mathcal{W}$ can be found in Figure~\ref{ch2:fig:KeywordLengthDistribution}.

\begin{figure}[h]
	\centering
	\includegraphics[width=0.98\linewidth]{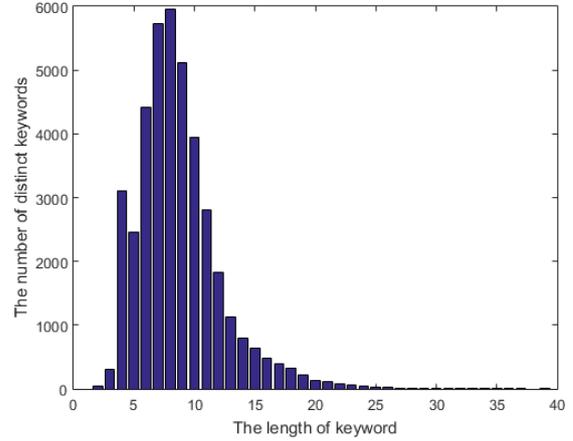}
	\captionsetup{justification=justified}
	\caption{The length distribution of a total of  40,205 distinct keywords in $\mathcal{W}$.}
	\label{ch2:fig:KeywordLengthDistribution}
\end{figure}

{
	In order to show the efficiency of our proposed substring-to-keyword scheme, we compare it with the schemes in \cite{chase2015, leontiadis2018storage}. Note that, in our experiment, we also use the schemes in \cite{chase2015, leontiadis2018storage} to support substring query on the dictionary string $t_\mathcal{W}$, which is transformed from the dictionary $\mathcal{W}$ by the method in Figure~\ref{ch2:fig:positionheapforkeywords}(a-b).}

\subsubsection{Data Outsourcing}
In this part, we consider the storage overhead and computational cost of data outsourcing phase.

{
	In general, given a dictionary string, our solution generates an encrypted position heap, \cite{chase2015} generates an encrypted suffix tree, and \cite{leontiadis2018storage} generates an encrypted suffix array as the index. Figure~\ref{ch2:fig:OutsourceStorage_filesize} and Figure~\ref{ch2:fig:OutsourcingTime_filesize} (the y-axis is log scale) depict the storage overhead and the runtime versus the size of dictionary (i.e., $m$) respectively, where $m$ varies from 5000 to 40000 keywords. The figures show that \cite{chase2015} consumes much more storage overhead and computation cost than \cite{leontiadis2018storage} and our solution in data outsourcing phase.}

\begin{figure}[h]
	\centering
	\includegraphics[width=0.98\linewidth]{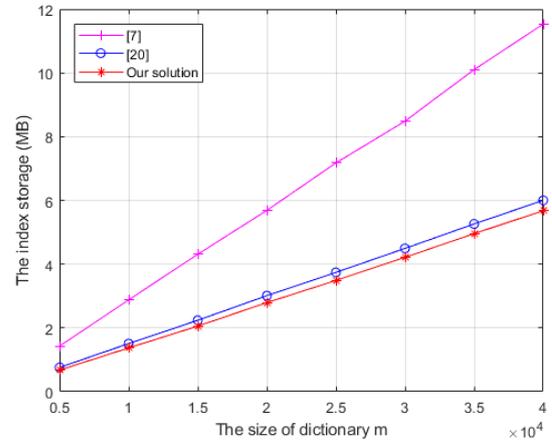}
	\captionsetup{justification=justified}
	\caption{The storage overhead of the data outsourcing versus the size of dictionary $m$.}
	\label{ch2:fig:OutsourceStorage_filesize}
\end{figure}

\begin{figure}[h]
	\centering
	\includegraphics[width=0.98\linewidth]{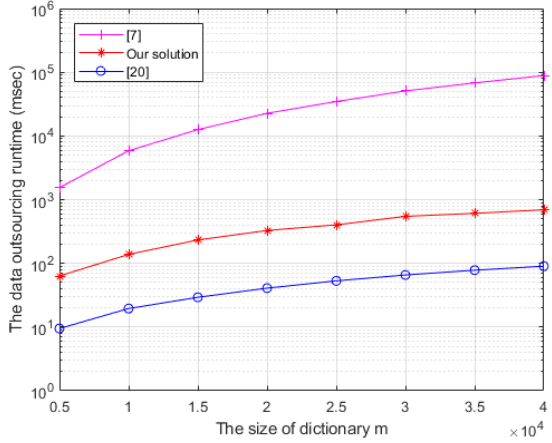}
	\captionsetup{justification=justified}
	\caption{The data outsourcing runtime versus the size of dictionary $m$.}
	\label{ch2:fig:OutsourcingTime_filesize}
\end{figure}

\subsubsection{Substring-to-keyword Query}
{In this part, we randomly choose queried substrings from the dictionary $\mathcal{W}$ and calculate their average queried time. Since the computational cost of substring-to-keyword query is limited by two factors: the size of dictionary (i.e., $m$) and the number of matching keywords (i.e., $d_s$), we analyze them separately in the following.}

\begin{figure}[h]
	\centering
	\includegraphics[width=0.98\linewidth]{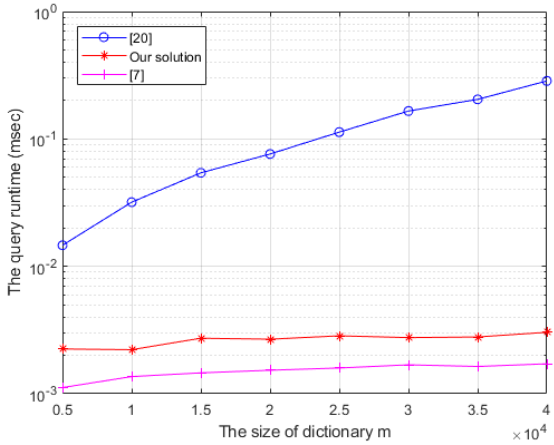}
	\captionsetup{justification=justified}
	\caption{{Substring-to-keyword query} runtime versus the size of dictionary $m$, where the number of matching keywords $d_s$ is 5.}
	\label{ch2:fig:QueryTime_filesize}
\end{figure}
\begin{figure}[h]
	\centering
	\includegraphics[width=0.98\linewidth]{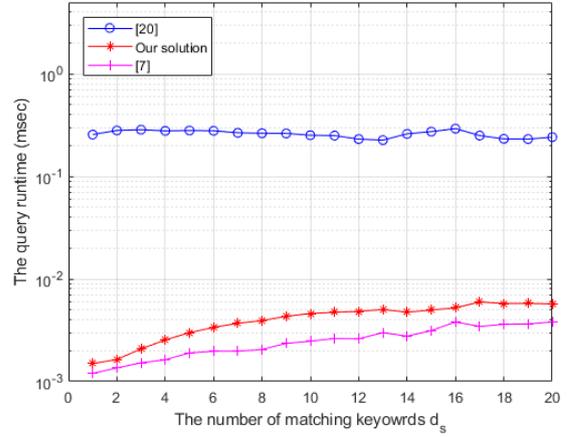}
	\captionsetup{justification=justified}
	\caption{{Substring-to-keyword query} runtime versus the number of matching keywords $d_s$, where $m = 40000$.}
	\label{ch2:fig:QueryTime_MatchingKeywordSize}
\end{figure}

{Figure~\ref{ch2:fig:QueryTime_filesize} (the y-axis is log scale) depicts the computational cost of the substring-to-keyword query versus the size of dictionary (i.e., $m$). This figure shows that the computational cost of our solution and \cite{chase2015} are not affected by $m$ when the number of matching keywords $d_s$ is fixed. However, the computational cost of \cite{leontiadis2018storage} increases linear with $m$ even if $d_s$ is fixed.
	
	Figure~\ref{ch2:fig:QueryTime_MatchingKeywordSize} (the y-axis is log scale) plots the runtime of the {substring-to-keyword query} versus the number of matching keywords $d_s$, in which $m$ is fixed to 40000. From this figure, we can see that the computational cost of three schemes are not affected too much by $d_s$. Meanwhile, our solution and \cite{chase2015} are significantly quicker than \cite{leontiadis2018storage}. For example, when $d_s = 20$, the computational cost of our solution and \cite{chase2015} are both about 0.004 ms, which is just about $1/60$ compared to \cite{leontiadis2018storage}.}

\subsubsection{Update}
{In this part, we consider the update (insertion / deletion) phase. Since there is no secure update method in \cite{chase2015, leontiadis2018storage}, we only test the update performance of our solution. Meanwhile, because the deletion operation in our solution is the same as the insertion operation, we just evaluate the computational cost of the insertion operation.}

\begin{figure}[h]
	\centering
	\includegraphics[width=0.98\linewidth]{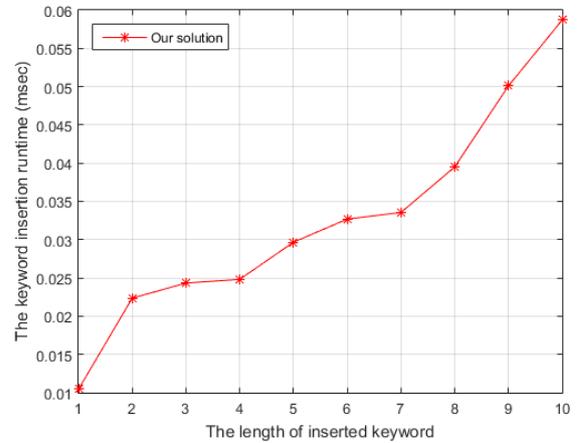}
	\captionsetup{justification=justified}
	\caption{Insertion runtime versus the size of inserted keyword, where the size of original dictionary $m$ is 5000.}
	\label{ch2:fig:InsertTime_keywordlen}
\end{figure}

{Figure~\ref{ch2:fig:InsertTime_keywordlen} plots the computational cost of the insertion versus the size of inserted keyword, in which the size of original dictionary $m$ is fixed to 5000. From this figure, we can see that the computational cost of our solution increase linearly with the size of inserted keyword.}

\section{Related Work}
\label{ch2:sec:related}
A searchable encryption scheme can be realized with optimal security via powerful cryptographic tools, such as Fully Homomorphic Encryption (FHE) \cite{Gentry09, Gentry10} and Oblivious Random Access Memory (ORAM) \cite{Ostrovsky90, GoldreichO96}. However, these tools are extraordinarily impractical. Another set of works utilize property-preserving encryption (PPE) \cite{BellareBO07, BoldyrevaCLO09, BoldyrevaCO11, YangXNSH18} to achieve searchable encryption, which encrypts messages in a way that inevitably leaks certain properties of the underlying message. For balancing the leakage and efficiency, many studies focus on Searchable Symmetric Encryption (SSE). Song et al. \cite{SongWP00} first used the symmetric encryption to facilitate keyword query over the encrypted data. Then, Curtmola et al. \cite{CurtmolaGKO06} gave a formal definition of SSE, and proposed an efficient SSE scheme. {However, their scheme cannot support update(insertion/deletion) operation. Later, Kamara et al. \cite{kamara2012} proposed the first dynamic SSE scheme, which uses a deletion array and a homomorphic encrypted pointer technique to securely update files. Unfortunately, due to the use of fully homomorphic encryption, the update efficiency is very low. In a more recent paper \cite{Cash2014}, Cash et al. described a simple dynamic inverted index based on \cite{CurtmolaGKO06}, which utilizes the data unlinkability of hash table to achieve secure insertion. Meanwhile, to prevent the file-injection attacks \cite{ZhangKP16}, many works \cite{bost2016, kim2017, zuo2018dynamic, zuo2019} focused on the forward security, which ensures that newly updated keywords cannot be related to previous queried results.
	
Nevertheless, these above works only can support the exact keyword query. If the queried keyword does not match a preset keyword, the query will fail. Fortunately, fuzzy query can deal with this problem as it can tolerate minor typos and formatting inconsistencies. Li et al. \cite{li2010} first proposed a fuzzy query scheme, which used an edit distance with a wildcard-based technique to construct fuzzy keyword sets. For instance, the set of $CAT$ with 1 edit distance is $\{CAT, *CAT, *AT, C*AT, C*T, CA*T, CA*, CAT*\}$. Then, Kuzu et al. \cite{kuzu2012} used LSH (Local Sensitive Hash) and Bloom filter to construct a similarity query scheme. Because an honest-but-curious server may only return a fraction of the results, Wang et al. \cite{wang2013} proposed a verifiable fuzzy query scheme that not only supports fuzzy query service, but also provides proof to verify whether the server returns all the queried results. However, these fuzzy query schemes only support single fuzzy keyword query and address problems of minor typos and formatting inconsistency, which can not be directly used to achieve substring-of-keyword query.
	
{In \cite{chase2015}, Melissa et al. designed a SSE scheme based on the suffix tree to support substring query. Although this scheme can be used to implement the substring-of-keyword query and allows for substring query in $O(|s| + d_s)$ time, its storage cost $O(m)$ has a big constant factor. The reason is that suffix tree only stores position data in leaf nodes and does not utilize the space of inner nodes effectively. This leads the number of nodes in suffix tree can be up to $2m$, where $m$ is the size of the dictionary. In order to reduce the storage cost as much as possible, Leontiadis et al.~\cite{leontiadis2018storage} leveraged Burrows Wheeler Transform (BWT) to build an auxiliary data structure called suffix array, which can achieve storage cost $O(m)$ with a lower constant factor. However, its query time is relatively large. Later, Mainardi et al.~\cite{mainardi2019privacy} optimizes the query algorithm in \cite{leontiadis2018storage} to achieve $O(|s| + d_s)$ at the cost of higher index space, i.e., $O(|\Sigma| \cdot m)$, where $|\Sigma|$ is the number of distinct characters in the dictionary. Although authors in this article considered datasets with small $|\Sigma|$ (e.g., DNA dataset), the $|\Sigma|$ can be large in practice.	In addition to suffix tree and suffix array, there are other auxiliary data structures can be used to support substring query.
In 2018, Florian et al.~\cite{hahn2018practical} designed an index consisting based on k-grams. When a user needs to perform a substring query, the cloud performs a conjunctive keyword query for all the k-grams of the queried substring. However, its query time is relatively large due to the computational cost of intersection operations in the conjunctive keyword query. In the same year, Tarik et al.~\cite{moataz2018substring} proposed a new substring query scheme based on the idea of letter orthogonalization, which allows testing of string membership by performing efficient inner product. Again, the disadvantage of this scheme comes its $O(m)$ query time. }

\section{Conclusion}
\label{ch2:sec:conclusions}
In this paper, we have proposed an efficient and privacy-preserving substring-of-keyword query scheme over cloud. Specifically, based on the position heap technique, we first designed a tree-based index to support substring-to-keyword query and then applied a PRF and a SKE to protect its privacy. After that, we proposed a novel substring-of-keyword query scheme, which contains two consecutive phases: a substring-to-keyword query that queries the keywords matching a given substring, and a keyword-to-file query that queries the files matching a keyword that the user is really interested. The proposed scheme is very suitable for many critical applications in practice such as Google search. Detailed security analysis and performance evaluation show that our proposed scheme is indeed privacy-preserving and efficient. In our future work, we will take more security properties into consideration, e.g., achieving forward and backward security.

\section{Acknowledgment}
This work is supported in part by NSERC Discovery Grants (no. Rgpin 04009), Natural Science Foundation of Zhejiang Province (grant no. LZ18F020003), National Natural Science Foundation of China (grant no. U1709217), and NSFC Grant (61871331).

\bibliographystyle{cas-model2-names}

\bibliography{paper}


\bio{figs/fyin}
\textbf{Fan Yin} received the B.S. degree in information
security from the Southwest Jiaotong University,
Chengdu, China, in 2012. He is currently
working toward the Ph.D. degree in information
and communication engineering, Southwest
Jiaotong University, and also a visiting student at Faculty of Computer Science,  University of New Brunswick, Canada.
His research interests include searchable encryption, privacy-preserving and security for cloud security and
network security.
\endbio

\bio{figs/rxlu}
	\textbf{Rongxing Lu} is currently an associate professor at the Faculty of Computer Science (FCS), University of New Brunswick (UNB), Canada. He is a Fellow of IEEE. His research interests include applied cryptography, privacy enhancing technologies, and IoT-Big Data security and privacy. He has published extensively in his areas of expertise, and was the recipient of 9 best (student) paper awards from some reputable journals and conferences. Currently, Dr. Lu serves as the Vice-Chair (Conferences) of IEEE ComSoc CIS-TC (Communications and Information Security Technical Committee). Dr. Lu is the Winner of 2016-17 Excellence in Teaching Award, FCS, UNB.
\endbio
~\\
~\\
~\\
~\\
~\\
~\\
~\\
\bio{figs/yzheng}
	\textbf{Yandong Zheng} received her M.S. degree from the Department of Computer Science, Beihang University, China, in 2017 and She is currently pursuing her Ph.D. degree in the Faculty of Computer Science, University of New Brunswick, Canada. Her research interest includes cloud computing security, big data privacy and applied privacy.
\endbio

~\\
\bio{figs/jshao}
	\textbf{Jun Shao} received the Ph.D. degree from the Department of Computer Science and Engineering, Shanghai Jiao Tong University, Shanghai, China, in 2008. He was a Postdoctoral Fellow with the School of Information Sciences and Technology, Pennsylvania State University, State College, PA, USA, from 2008 to 2010. He is currently a Professor with the School of Computer Science and Information Engineering, Zhejiang Gongshang University, Hangzhou, China. His current research interests include network security and applied cryptography.
\endbio

\bio{figs/xueyang}
	\textbf{Xue Yang} received the Ph.D degree in Information and Communication Engineering from Southwest Jiaotong University, Chengdu, China, in 2019. She was a visiting student at the Faculty of Computer Science, University of New Brunswick, Canada, from 2017 to 2018. She is currently a Postdoctoral Fellow in the Tsinghua Shenzhen International Graduate School, Tsinghua University, China. Her research interests include big data security and privacy, applied cryptography and federated learning.
\endbio

~\\

\bio{figs/xhtang}
	\textbf{Xiaohu Tang} received the B.S. degree in applied mathematics from Northwest Polytechnic University, Xi'an, China, in 1992, the M.S. degree in applied mathematics from Sichuan University, Chengdu, China, in 1995, and the Ph.D. degree in electronic engineering from Southwest Jiaotong University, Chengdu, in 2001. 	
	From 2003 to 2004, he was a Research Associate with the Department of Electrical and Electronic Engineering, The Hong Kong University of Science and Technology. From 2007 to 2008, he was a Visiting Professor with the University of Ulm, Germany. Since 2001, he has been with the School of Information Science and Technology, Southwest Jiaotong University, where he is currently a Professor. His research interests include coding theory, network security, distributed storage, and information processing for big data. 	
	Dr. Tang was a recipient of the National excellent Doctoral Dissertation Award in 2003 (China), the Humboldt Research Fellowship in 2007 (Germany), and the Outstanding Young Scientist Award by NSFC in 2013 (China). He served as an Associate Editor for several journals, including the IEEE TRANSACTIONS ON INFORMATION THEORY and IEICE Transactions on Fundamentals, and served for a number of technical program committees of conferences.
\endbio

\end{document}